\documentclass[11pt]{article}

\usepackage{amsfonts}
\usepackage{amsmath}
\usepackage{amssymb}

\usepackage{tikz}

\usepackage{subfigure}
\usepackage{captcont}

\usepackage[final,letterpaper]{hyperref}

\hypersetup{
pdfauthor = {Xiaoyang Gu, Jack H. Lutz, Elvira Mayordomo},
pdftitle = {Curves That Must Be Retraced},
pdfsubject = {Curves That Must Be Retraced},
pdfkeywords = {computable curve, rectifiable curve, computability, computable analysis},
pdfcreator = {LaTeX with hyperref package},
pdfproducer = {dvips + ps2pdf}}

\usepackage{amsfonts}
\usepackage{amsmath}
\usepackage{amssymb}
\usepackage{enumerate}

\newcommand{\PROOF}{\begin{proof}}

\newcommand{\QED}{\end{proof}}

\newcommand{\range}{\mathrm{range}}

\newcommand{\myset}[2]{ \left\{ #1 \;\left|\; #2 \right. \right\} }

\newcommand{\limn}{\lim\limits_{n\to\infty}}
\newcommand{\liminfn}{\liminf\limits_{n\to\infty}}

\newcommand{\N}{\mathbb{N}}
\newcommand{\Z}{\mathbb{Z}}
\newcommand{\Q}{\mathbb{Q}}
\newcommand{\R}{\mathbb{R}}

\newcommand{\K}{{\mathrm{K}}}

\newcommand{\CH}{\mathcal{H}}

\usepackage{amsthm}

\newtheorem{theorem}{Theorem}[section]
\newtheorem{corollary}[theorem]{Corollary}
\newtheorem{lemma}[theorem]{Lemma}
\newtheorem{proposition}[theorem]{Proposition}
\newtheorem*{claim}{Claim}

\newtheorem{observation}[theorem]{Observation}

\theoremstyle{definition}

\newtheorem*{definition}{Definition}
\newtheorem*{defn}{Definition}

\newtheorem{example}[theorem]{Example}

\newtheorem*{example*}{Example}
\newtheorem*{examples*}{Examples}

\newtheorem{construction}[theorem]{Construction}

\newtheorem*{ack}{Acknowledgment}

\theoremstyle{remark}

\numberwithin{equation}{section}
\numberwithin{figure}{section}

\renewcommand{\include}{\input}

\newcommand{\card}[1]{\lvert #1 \rvert}
\newcommand{\hd} {d_{\mathrm{H}} }

\begin{document}

\title{ {\bf
Curves That Must Be Retraced}
}
\author{
Xiaoyang Gu
\footnote{Department of Computer Science, Iowa State University,
Ames, IA 50011, USA. Email: xiaoyang@cs.iastate.edu}
\footnotemark[3]
\and
Jack H. Lutz\footnote{Department of Computer Science, Iowa State University,
Ames, IA 50011, USA.  Email: lutz@cs.iastate.edu}
\footnote{Research supported in part by National Science Foundation
Grant 0344187, 0652569, and 0728806.} \footnotemark[5]
\and
Elvira Mayordomo
\footnote{Departamento de Inform\'atica e Ingenier\'ia de Sistemas,
Universidad de Zaragoza, 50018 Zaragoza, Spain.  Email: elvira@unizar.es}
\footnote{Research supported in part by the Spanish Ministry of Education and Science (MEC) and the European Regional Development Fund (ERDF) under project TIN2005-08832-C03-02.}
\footnote{Part of this author's research was performed during a visit at Iowa State University, supported by Spanish Government (Secretar\'{\i}a de Estado de Universidades e Investigaci\'on del Ministerio de Educaci\'on y Ciencia) grant for research stays PR2007-0368. }
}

\date{}

\maketitle

\begin{abstract}
We exhibit a polynomial time computable plane curve ${\bf \Gamma}$
that has finite length, does not intersect itself, and
is smooth except at one endpoint, but has the following property.
For every computable parametrization $f$ of ${\bf\Gamma}$ and
every positive integer $m$, there is some positive-length
subcurve of ${\bf\Gamma}$ that $f$ retraces at least $m$ times.
In contrast, every computable curve of finite length
that does not intersect itself has a constant-speed
(hence non-retracing) parametrization that is computable
relative to the halting problem.
\end{abstract}

\begin{section}{Introduction}\label{se:1}
A curve is a mathematical model of the path
of a particle undergoing continuous motion.
Specifically, in a Euclidean space $\R^n$, a curve
is the range $\Gamma$ of a continuous function
$f:[a,b]\rightarrow \R^n$ for some $a<b$. The function
$f$, called a {\em parametrization} of $\Gamma$, clearly
contains more information than the pointset $\Gamma$,
namely, the precise manner in which the
particle ``traces'' the points $f(t)\in\Gamma$ as
$t$, which is often considered a time parameter,
varies from $a$ to $b$.	When the particle's motion is algorithmically governed, the parametrization must be computable (as a function on the reals, see below).

This paper shows that the geometry of a
curve $\Gamma$ may force every
{\em computable} parametrization $f$ of $\Gamma$ to retrace
various parts of its path (i.e., ``go back and forth
along $\Gamma$'') many times, even when $\Gamma$ is an
efficiently computable, smooth, finite-length curve
that does not intersect itself. In fact, our main
theorem exhibits a plane curve ${\bf\Gamma}\subseteq \R^2$ with
the following properties.
\begin{enumerate}
\item
${\bf \Gamma}$ is {\em simple}, i.e., it does not
intersect itself.
\item
${\bf \Gamma}$ is {\em rectifiable}, i.e., it has finite
length.
\item
${\bf \Gamma}$ is {\em smooth except at one endpoint},
i.e., ${\bf \Gamma}$ has a tangent at every interior
point and a $1$-sided tangent at one endpoint,
and these tangents vary continuously along ${\bf \Gamma}$.
\item\label{pr:1_4}
${\bf \Gamma}$ is {\em polynomial time computable} in the strong
sense that there is a polynomial time computable
position function $\vec{s}:[0,1]\rightarrow\R^2$ such that the
velocity function $\vec{v} =\vec{s}'$ and the acceleration
function $\vec{a} =\vec{v}'$ are polynomial time computable;
the total distance traversed by $\vec{s}$ is finite;
and $\vec{s}$ parametrizes ${\bf \Gamma}$, i.e., $\range(\vec{s})={\bf \Gamma}$.
\item\label{pr:1_5}
${\bf \Gamma}$ {\em must be retraced} in the sense that
every parametrization $f:[a,b]\rightarrow\R^2$ of ${\bf \Gamma}$ that is computable
in {\em any} amount of time has the following property.
For every positive integer $m$, there exist disjoint,
closed subintervals $I_0, \dots, I_m$ of $[a,b]$
such that the curve $\Gamma_0=f(I_0)$ has positive
length and $f(I_i)=\Gamma_0$ for all $1\leq i\leq m$.
(Hence $f$ retraces $\Gamma_0$ at least $m$ times.)
\end{enumerate}

The terms ``computable'' and ``polynomial
time computable'' in properties \ref{pr:1_4} and \ref{pr:1_5} above
refer to the ``bit-computability'' model of
computation on reals formulated in the 1950s
by Grzegorczyk \cite{Grze55} and Lacombe \cite{Laco55},
extended to feasible computability in the 1980s by Ko and
Friedman \cite{KoFri82} and Kreitz and Weihrauch \cite{KreWei82},
and exposited in the recent paper by Braverman
and Cook \cite{BraCoo06} and the monographs
\cite{PouRic89,Ko91,Weihr00,BraYam08}. As will be shown here,
condition \ref{pr:1_4} also implies that the pointset ${\bf \Gamma}$
is polynomial time computable in the sense of
Brattka and Weihrauch \cite{BraWei99}. (See also
\cite{Weihr00,Brav05,BraCoo06}.)

A fundamental and useful theorem of
classical analysis states that every simple,
rectifiable curve $\Gamma$ has a {\em normalized constant-speed
parametrization}, which is a one-to-one parametrization
$f:[0,1]\rightarrow\R^n$ of $\Gamma$ with the property that
$f([0,t])$ has arclength $tL$ for all $0\le t\leq 1$,
where $L$ is the length of $\Gamma$. (A simple,
rectifiable curve $\Gamma$ has exactly two such
parametrizations, one in each direction, and
standard terminology calls either of these {\em the} normalized
constant-speed parametrization $f:[0,1]\rightarrow\R^n$ of $\Gamma$.
The constant-speed  parametrization is also called
the {\em parametrization by arclength} when it is
reformulated as a function $f:[0,L]\rightarrow\R^n$ that
moves with constant speed $1$ along $\Gamma$.) Since the
constant-speed parametrization does not retrace
any part of the curve, our main theorem
implies that this classical theorem is not entirely
constructive. Even when a simple, rectifiable
curve has an efficiently computable parametrization,
the constant-speed parametrization need not be
computable.

In addition to our main theorem, we
prove that every simple, rectifiable curve $\Gamma$
in $\R^n$ with a computable parametrization has
the following two properties.
\begin{enumerate}[I.]
\item\label{pr:1_6}
The length of $\Gamma$ is lower  semicomputable.
\item\label{pr:1_7}
The constant-speed parametrization of $\Gamma$
is computable relative to the length of $\Gamma$.
\end{enumerate}

These two things are not hard to prove if the computable parametrization is one-to-one,
(in fact, they follow from results of M\"uller and Zhao \cite{Muller:JAG} in this case)
but our results hold even when the computable
parametrization retraces portions of the curve many times.

Taken together, \ref{pr:1_6} and \ref{pr:1_7} have the following two consequences.
\begin{enumerate}
\item
The curve $\mathbf{\Gamma}$ of our main theorem
has a finite length that is lower semi-computable
but not computable. (The existence of polynomial-time
computable curves with this property was
first proven by Ko \cite{Ko95}.)
\item
Every simple, rectifiable curve $\Gamma$ in $\R^n$
with a computable parametrization has a
constant-speed parametrization that is $\Delta_2^0$-computable, i.e.,
computable relative to the
halting problem. Hence, the existence of
a constant-speed parametrization, while not
entirely constructive, is constructive relative
to the halting problem.
\end{enumerate}
\end{section}

\begin{section}{Length, Computability, and Complexity of Curves}\label{se:2}

In this section we summarize basic terminology and facts about curves.
 As we use the terms here, a {\sl curve} is the range $\Gamma$ of a continuous function $f: [a, b] \to \R^n$ for some $a<b$. The function $f$ is called a {\sl parametrization} of $\Gamma$. Each curve clearly has infinitely many parametrizations.

 A curve is {\sl simple} if it has a parametrization that is one-to-one, i.e., the curve ``does not intersect itself''. The length of a simple curve $\Gamma$ is defined as follows. Let $f: [a, b]\stackrel{1-1}{\to}\R^n$ be a one-to-one parametrization of $\Gamma$. For each {\sl disection} $\vec{t}$ of $[a,b]$, i.e., each tuple $\vec{t}=(t_0, \ldots, t_m)$ with $a=t_0< t_1< \ldots <t_m=b$, define the $f$-$\vec{t}${\sl -approximate length} of $\Gamma$ to be
 \[\mathcal{L}^f_{\vec{t}}(\Gamma)
 =\sum_{i=0}^{m-1}|f(t_{i+1})-f(t_i)|.\]
Then the {\sl length} of $\Gamma$ is
 \[\mathcal{L}(\Gamma)=\sup_{\vec{t}}
 \mathcal{L}^f_{\vec{t}}(\Gamma),\]
 where the supremum is taken over all dissections $\vec{t}$ of $[a,b]$. It is easy to show that $\mathcal{L}(\Gamma)$ does not depend on the choice of the one-to-one parametrization $f$, i.e. that the length is an intrinsic property of the pointset $\Gamma$.

 In sections \ref{se:4} and \ref{se:5}  of this paper we use a more general notion of length, namely, the 1-dimensional Hausdorff measure $\mathcal{H}^1(\Gamma)$, which is defined for {\sl every} set $\Gamma\subseteq\R^n$. We refer the reader to \cite{Falc03} or the appendix for the definition of  $\mathcal{H}^1(\Gamma)$. It is well known that  $\mathcal{H}^1(\Gamma)=\mathcal{L}(\Gamma)$
 holds for every simple curve $\Gamma$.

 A curve $\Gamma$ is {\sl rectifiable\/}, or {\sl has finite length} if $\mathcal{L}(\Gamma)<\infty$. In sections \ref{se:4} and \ref{se:5} we use the notation $\mathcal{RC}$ for the set of all rectifiable simple curves.
\begin{definition}
Let $f:[a,b]\rightarrow\R^n$ be continuous.
\begin{enumerate}
\item
For $m\in\Z^+$, $f$ has $m$-{\em fold retracing} if
there exist disjoint, closed subintervals
$I_0,\dots, I_m$ of $[a,b]$ such that the curve
$\Gamma_0=f(I_0)$ has positive length and $f(I_i)=\Gamma_0$
for all $1\leq i\leq m$.
\item
$f$ is {\em non}-{\em retracing} if $f$ does not have $1$-fold retracing.
\item
$f$ has {\em bounded retracing} if there exists
$m\in\Z^+$ such that $f$ does not have $m$-fold
retracing.
\item
$f$ has {\em unbounded retracing} if $f$ does
not have bounded retracing, i.e., if $f$ has
$m$-fold retracing for all $m\in\Z^+$.
\end{enumerate}
\end{definition}
 We now review the notions of computability and complexity of a real-valued function. An {\sl oracle} for a real number $t$ is any function $O_t: \N\to\Q$ with the property that $|O_t(s)-t|\le 2^{-s}$ holds for all $s\in\N$. A function $f: [a,b]\to \R^n$ is {\sl computable} if there is an oracle Turing machine $M$ with the following property. For every $t\in [a,b]$ and every precision parameter $r\in\N$, if $M$ is given $r$ as input and {\sl any} oracle $O_t$ for $t$ as its oracle, then $M$ outputs a rational point $M^{O_t}(r)\in\Q^n$ such that $|M^{O_t}(r)-f(t)|\le 2^{-r}$.
 A function $f: [a,b]\to \R^n$ is {\sl computable in polynomial time} if there is an oracle machine $M$ that does this in time polynomial in $r+l$, where $l$ is the maximum length of the query responses provided by the oracle.

An {\em oracle} for a function $f:[a,b]\rightarrow\R^n$
is any function $\mathcal{O}_f:([a,b]\cap \Q)\times\N\rightarrow\Q^n$
with the property that $|\mathcal{O}_f(q,r)-f(q)|\leq 2^{-r}$
holds for all $q\in[a,b]\cap \Q$ and $r\in\N$. A
decision problem $A$ is {\em Turing reducible } to
a function $f:[a,b]\rightarrow\R^n$, and we write
$A\leq_\mathrm{T} f$, if there is an oracle Turing
machine $M$ such that, for every oracle
$\mathcal{O}_f$ for $f$, $M^{\mathcal{O}_f}$ decides $A$. It is easy
to see that, if  $f$ is computable, then
$A\leq_\mathrm{T} f$ if and only if $A$ is decidable.

 A curve is {\sl computable} if it has a parametrization $f:[a,b]\rightarrow\R^n$, where $a,b\in\Q$ and $f$
is computable.  A curve is {\sl computable in polynomial time}
if it has a parametrization that is computable in polynomial
time.

\end{section} 
\begin{section}{An Efficiently Computable Curve That Must Be Retraced}\label{se:3}
This section presents our main theorem, which is
the existence of a smooth, rectifiable, simple plane
curve ${\bf \Gamma}$ that is parametrizable in polynomial time
but not computably parametrizable in any amount
of time without unbounded retracing. We
begin with a precise construction of the
curve ${\bf \Gamma}$, followed by a brief intuitive
discussion of this construction. The rest of
the section is devoted to proving that ${\bf \Gamma}$
has the desired properties.

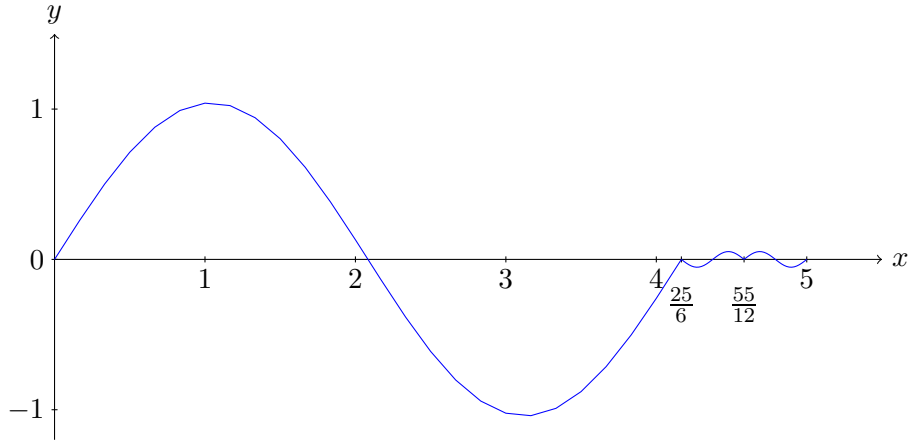
\begin{figure}[htbp]
\begin{center}
\begin{tikzpicture}[scale=2]
\draw[->](0,-1.2)--(0,1.5) node[above]{$y$};
\draw[->](0,0)--(5.5,0) node[right]{$x$};
\foreach \x in {-1,...,1}
    {\draw  (-0.02,\x) -- (0.02,\x);
     \draw (0,\x) node[left] {$\x$};}
\foreach \x in {1,...,5}
    {\draw (\x, -0.02) -- (\x, 0.02);
     \draw (\x, 0) node[below] {$\x$};}
\draw (25/6,-0.02) -- (25/6,0.02);
\draw (25/6,-0.3) node {$\frac{25}{6}$};
\draw (55/12,-0.02) -- (55/12,0.02);
\draw (55/12,-0.3) node {$\frac{55}{12}$};
\draw[color=blue,domain=0:25/6] plot (\x,{(25/6/4)*sin(2*pi*(\x)*6/25  r)});
\draw[color=blue,domain=50/12:55/12] plot (\x,{(5/24/4)*(-1)*sin(2*pi*(\x-50/12)*12/5  r)});
\draw[color=blue,domain=55/12:5] plot (\x,{(5/24/4)*sin(2*pi*(\x-55/12)*12/5  r)});
\end{tikzpicture}
\caption{$\psi_{0,5,1}$}
\label{fig:3_1}
\end{center}
\end{figure}


\begin{construction}\label{con:3_1}
\begin{enumerate}[(1)]
\item\label{it:3_1}
For each $a,b\in\R$ with $a<b$, define the functions
$\varphi_{a,b}, \xi_{a,b}:[a,b]\rightarrow\R$ by
\[\varphi_{a,b}(t)=\frac{b-a}{4}\sin \frac{2\pi (t-a)}{b-a}\]
and
\[
\xi_{a,b}(t)=
\begin{cases}
-\varphi_{a,\frac{a+b}{2}}(t)& \text{if } a\leq t\leq \frac{a+b}{2}\\
\varphi_{\frac{a+b}{2},b}(t)&\text{if }\frac{a+b}{2}\leq t\leq b.
\end{cases}
\]
\item\label{it:3_2}
For each $a,b\in\R$ with $a<b$ and each positive
integer $n$, define the function $\psi_{a,b,n}:[a,b]\rightarrow\R$ by
\[
\psi_{a,b,n}(t) = \begin{cases}
\varphi_{a,d_0}(t)&\text{if }a\leq t\leq d_0\\
\xi_{d_{i-1},d_i}(t)&\text{if }d_{i-1}\leq t\leq d_i,
\end{cases}
\]
where
\[d_i=\frac{a+5b}{6}+i\frac{b-a}{ 6n}\]
for $0\leq i\leq n$. (See Figure \ref{fig:3_1}.)
\item\label{it:3_3}
Fix a standard enumeration $M_1, M_2, \dots$
of (deterministic) Turing machines that take
positive integer inputs. For each positive
integer $n$, let $\tau(n)$ denote the number of
steps executed by $M_n$ on input $n$.
It is well known that the {\em diagonal halting problem}
\[K=\myset{n\in\Z^+ }{ \tau(n)<\infty }\]
is undecidable.
\item\label{it:3_4}
Define the horizontal and
vertical acceleration functions $a_x, a_y:[0,1]\rightarrow\R$
as follows. For each $n\in\N$, let
\[t_n = \int_0^n e^{-x}dx = 1- e^{-n},\]
noting that $t_0=0$ and that $t_n$ converges
monotonically to $1$ as $n\rightarrow\infty$.
Also, for each $n\in\Z^+$, let
\[t_n^- =\frac{t_{n-1} + 4t_n}{5},\;\; t_n^+ =\frac{6t_n -t_{n-1}}{5},\]
noting that these are symmetric about $t_n$
and that $t_n^+\leq t_{n+1}^-$.

\begin{enumerate}[(i)]
\item
For $0\leq t\leq 1$, let
\[a_x(t)=
\begin{cases}
-2^{-(n+\tau(n))}\xi_{t_n^-,t_n^+}(t)&\text{if }t_n^-\leq t <t_n^+\\
0&\text{if no such $n$ exists},
\end{cases}
\]
where $2^{-\infty}=0$.
\item
For $0\leq t < 1$, let
\[a_y(t)=\psi_{t_{n-1}, t_n, n}(t),\]
where $n$ is the unique positive integer
such that $t_{n-1}\leq t<t_n$.
\item
Let $a_y(1)=0$.
\end{enumerate}
\item\label{it:3_5}
Define the horizontal and vertical velocity
and position functions $v_x, v_y, s_x, s_y:[0,1]\rightarrow \R$
by
\begin{align*}
v_x(t)=\int_0^t a_x(\theta)d\theta,\;\;\;  &v_y(t)=\int_0^t a_y(\theta) d\theta,\\
s_x(t)=\int_0^t v_x(\theta)d\theta,\;\;\;  & s_y(t)=\int_0^t v_y(\theta)d\theta.
\end{align*}
\item\label{it:3_6}
Define the vector acceleration, velocity,
and position functions $\vec{a},\vec{v},\vec{s}:[0,1]\rightarrow \R^2$
by
\begin{align*}
&\vec{a}(t)=(a_x(t),a_y(t)),\\
&\vec{v}(t)=(v_x(t), v_y(t)),\\
&\vec{s}(t)=(s_x(t),s_y(t)).
\end{align*}
\item\label{it:3_7}
Let ${\bf \Gamma} =\range (\vec{s})$.
\end{enumerate}
\end{construction}

Intuitively, a particle at rest at time
$t=a$ and moving with acceleration given by
the function $\varphi_{a,b}$ moves forward, with velocity
increasing to a maximum at time $t=\frac{a+b}{2}$
and then decreasing back to $0$ at time $t=b$.
The vertical acceleration function $a_y$, together
with the initial conditions $v_y(0)=s_y(0)=0$
implied by (\ref{it:3_5}), thus causes a particle
to move generally upward (i.e., $s_y(t_0)<s_y(t_1)<\cdots$),
coming to momentary rests at times $t_1, t_2, t_3,\dots$.
Between two consecutive such stopping times
$t_{n-1}$ and $t_n$, the particle's vertical
acceleration is controlled by the function
$\psi_{t_{n-1},t_n,n}$. This function causes the
particle's vertical motion to do the following
between times $t_{n-1}$ and $t_n$.
\begin{enumerate}[(i)]
\item\label{it:3_2_1}
From time $t_{n-1}$ to time $\frac{t_{n-1}+5t_n}{6}$,
move upward from elevation $s_y(t_{n-1})$
to elevation $s_y(t_n)$.
\item\label{it:3_2_2}
From time $\frac{t_{n-1}+5t_n}{6}$ to time $t_n$,
make $n$ round trips to a lower elevation
$s\in( s_y(t_{n-1}), s_y(t_n))$.
\end{enumerate}
In the meantime, the horizontal acceleration
function $a_x$, together with the initial conditions
$v_x(0)=s_x(0)=0$ implied by (\ref{it:3_5}), ensure that
the particle remains on or near the
$y$-axis. The deviations from the $y$-axis
are simply described: The particle moves
to the right from time $\frac{t_{n-1}+4t_n}{5}$ through
the completion of the $n$ round trips
described in (\ref{it:3_2_2}) above and then moves
to the $y$-axis between times $t_n$ and $\frac{6t_n-t_{n-1}}{5}$.
The amount of lateral motion here is regulated
by the coefficient $2^{-(n+\tau(n))}$. If $\tau(n)=\infty$,
then there is no lateral motion, and the $n$
round trips in (\ref{it:3_2_2}) are retracings of the
particle's path. If $\tau(n)<\infty$, then these
$n$ round trips are ``forward'' motion along
a curvy part of ${\bf \Gamma}$. In fact, ${\bf \Gamma}$ contains
points of arbitrarily high curvature,
but the particle's motion is kinematically
realistic in the sense that the acceleration
vector $\vec{a}(t)$ is polynomial time computable,
hence continuous and bounded on the
interval $[0,1]$. Figure \ref{fig:3_2} illustrates
the path of the particle from time $t_{n-1}$ to $t_{n+1}$
with $n=1$ and hypothetical (model dependent!)
values $\tau(1)=1$ and $\tau(2)=2$.

\begin{figure}[htbp]
\begin{center}
\begin{tikzpicture}[scale=1.75,y=60]
\foreach \x in {1,2}
    {\draw (\x, 1-0.02) -- (\x, 1+0.02);}
\draw[gray,->](0,0.95)--(0,2.25) node[above]{$y$};
\draw[gray,->](-0.1,1)--(3,1) node[right]{$x$};
\draw[color=black,domain=3:4,->] plot (0,{0.5*25/6/2/pi*\x+0.5*25/6/2/pi*25/6/2/pi*(-1)*sin(2*pi*\x*6/25  r)});
\draw[color=black,domain=4:25/6,->] plot ({60*((1/4)*(0.5/pi)*(\x-4)+(-1/4)*(0.5/pi/2/pi)*sin(2*pi*(\x-4)  r))},{0.5*25/6/2/pi*\x+0.5*25/6/2/pi*25/6/2/pi*(-1)*sin(2*pi*\x*6/25  r)});
\draw[color=black,domain=25/6:55/12,->] plot ({60*((1/4)*(0.5/pi)*(\x-4)+(-1/4)*(0.5/pi/2/pi)*sin(2*pi*(\x-4)  r))},{0.5*25/6/2/pi*25/6-0.25*(\x-50/12)+0.25/2/pi/12*5*sin(2*pi*(\x-50/12)*12/5  r)});
\draw[color=black,domain=55/12:5,->] plot ({60*((1/4)*(0.5/pi)*(\x-4)+(-1/4)*(0.5/pi/2/pi)*sin(2*pi*(\x-4)  r))},{0.5*25/6/2/pi*25/6-0.25*5/12+0.25*(\x-55/12)+0.25*(-1)/2/pi/12*5*sin(2*pi*(\x-55/12)*12/5  r)});
\draw[color=black,domain=5:6,->] plot ({60*((-1/4)*(0.5/pi)*(\x-5)+(1/4)*(0.5/pi/2/pi)*sin(2*pi*(\x-5)  r)+(1/4)*(0.5/pi)*(5-4)+(-1/4)*(0.5/pi/2/pi)*sin(2*pi*(5-4)  r))},{0.5*25/6/2/pi*25/6+0.4*25/6/2/pi*(\x-5)+0.4*(-1)*25/6/2/pi/2/pi/12*25*sin(2*pi*(\x-5)*12/25  r)});
\draw[color=black,domain=6:7,->] plot (0,{0.5*25/6/2/pi*25/6+0.4*25/6/2/pi*(\x-5)+0.4*(-1)*25/6/2/pi/2/pi/12*25*sin(2*pi*(\x-5)*12/25  r)});
\draw[color=black,domain=7:85/12,->] plot ({200*((1/8)*(0.5/pi)*(\x-7)+(-1/8)*(0.5/pi/2/pi*0.5)*sin(2*pi*(\x-7)/0.5  r))},{0.5*25/6/2/pi*25/6+0.4*25/6/2/pi*(\x-5)+0.4*(-1)*25/6/2/pi/2/pi/12*25*sin(2*pi*(\x-5)*12/25  r)});
\draw[color=black,domain=85/12:172.5/24,->] plot ({200*((1/8)*(0.5/pi)*(\x-7)+(-1/8)*(0.5/pi/2/pi*0.5)*sin(2*pi*(\x-7)/0.5  r))},{0.5*25/6/2/pi*25/6+0.4*25/6/2/pi*25/12-0.2*(\x-85/12)+0.2/2/pi/24*2.5*sin(2*pi*(\x-85/12)*24/2.5  r)});
\draw[color=black,domain=172.5/24:175/24,->] plot ({200*((1/8)*(0.5/pi)*(\x-7)+(-1/8)*(0.5/pi/2/pi*0.5)*sin(2*pi*(\x-7)/0.5  r))},{0.5*25/6/2/pi*25/6+0.4*25/6/2/pi*25/12-0.2*2.5/24+0.2*(\x-172.5/24)-0.2/2/pi/24*2.5*sin(2*pi*(\x-172.5/24)*24/2.5  r)});
\draw[color=black,domain=175/24:177.5/24,->] plot ({200*((1/8)*(0.5/pi)*(\x-7)+(-1/8)*(0.5/pi/2/pi*0.5)*sin(2*pi*(\x-7)/0.5  r))},{0.5*25/6/2/pi*25/6+0.4*25/6/2/pi*25/12-0.2*(\x-175/24)+0.2/2/pi/24*2.5*sin(2*pi*(\x-175/24)*24/2.5  r)});
\draw[color=black,domain=177.5/24:180/24] plot ({200*((1/8)*(0.5/pi)*(\x-7)+(-1/8)*(0.5/pi/2/pi*0.5)*sin(2*pi*(\x-7)/0.5  r))},{0.5*25/6/2/pi*25/6+0.4*25/6/2/pi*25/12-0.2*2.5/24+0.2*(\x-177.5/24)-0.2/2/pi/24*2.5*sin(2*pi*(\x-177.5/24)*24/2.5  r)});
\end{tikzpicture}
\caption{Example of $\vec{s}(t)$ from $t_0$ to $t_2$}
\label{fig:3_2}
\end{center}
\end{figure}
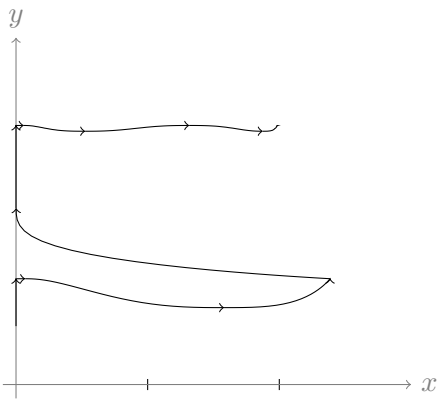

The rest of this section is devoted to
proving the following theorem concerning the
curve ${\bf \Gamma}$.

\begin{theorem}(main theorem).\label{th:3_2}
Let $\vec{a}, \vec{v}, \vec{s}$,
and ${\bf \Gamma}$ be as in Construction \ref{con:3_1}.
\begin{enumerate}
\item
The functions $\vec{a}, \vec{v}$, and $\vec{s}$ are Lipschitz
and computable in polynomial time, hence
continuous and bounded.
\item
The total length, including retracings, of
the parametrization $\vec{s}$ of ${\bf \Gamma}$ is finite and
computable in polynomial time.
\item
The curve ${\bf \Gamma}$ is simple, rectifiable,
and smooth except at one endpoint.
\item
Every computable parametrization
$f:[a,b]\rightarrow\R^2$ of ${\bf \Gamma}$ has {\em unbounded
retracing}.
\end{enumerate}
\end{theorem}

For the remainder of this section, we
use the notation of Construction \ref{con:3_1}.

The following two observations facilitate our analysis
of the curve ${\bf \Gamma}$. The proofs are routine calculations.

\begin{observation}\label{ob:3_3}
For all $n\in\Z^+$, if we write
\[d_i^{(n)} =\frac{t_{n-1} + 5t_n}{6}  + i\frac{t_n-t_{n-1}}{6n}\]
and
\[e_i^{(n)}=d_i^{(n)} +\frac{t_n-t_{n-1}}{12n}\]
for all $0\leq i<n$, then
\[t_{n-1} <t_n^-<d_0^{(n)} <e_0^{(n)}<d_1^{(n)}<e_1^{(n)}
<\cdots < d_{n-1}^{(n)} < e_{n-1}^{(n)} <t_n<t_n^+<t_{n+1}^-.\]
\end{observation}
\begin{observation}\label{ob:3_4}
For all $a,b\in\R$ with $a<b$,
\[\int_a^b\int_a^t \varphi_{a,b}(\theta)d\theta dt =\frac{(b-a)^3}{8\pi}.\]
\end{observation}

We now proceed with a quantitative analysis
of the geometry of ${\bf \Gamma}$. We begin with the
horizontal component of $\vec{s}$.

\begin{lemma}\label{lm:3_5}
\begin{enumerate}
\item
For all $t\in[0,1] -\bigcup_{n\in K} (t_n^-,t_n^+)$,
$v_x(t)=s_x(t)=0$.
\item
For all $n\in K$ and $t\in(t_n^-, t_n)$ , $v_x(t)>0$.
\item
For all $n\in K$ and $t\in(t_n, t_n^+)$, $v_x(t)<0$.
\item
For all $n\in\Z^+$, $s_x(t_n)=\frac{(e-1)^3}{1000\pi e^{3n}}  2^{-(n+\tau(n))}$.
\item
$s_x(1)=0$.
\end{enumerate}
\end{lemma}
\begin{proof}
Parts 1-3 are routine by inspection and
induction. For $n\in\Z^+$, Observation \ref{ob:3_4} tells us
that
\begin{align*}
s_x(t_n)&=\frac{(t_n - t_n^-)^3}{8\pi}2^{-(n+\tau(n))}\\
&=\frac{(\tfrac{1}{5}(t_n-t_{n-1}))^3 }{8\pi}2^{-(n+\tau(n))}\\
&=\frac{(\tfrac{1}{5}( (e-1)e^{-n} ))^3 }{8\pi}2^{-(n+\tau(n))}\\
&=\frac{(e-1)^3}{1000 \pi e^{3n}}2^{-(n+\tau(n))}
\end{align*}
so 4 holds. This implies that $s_x(t_n)\rightarrow 0$
as $n\rightarrow\infty$, whence 5 follows from 1,2, and 3.
\end{proof}

The following lemma analyzes the vertical component
of $\vec{s}$. We use
the notation of Observation \ref{ob:3_3},
with the additional proviso that $d_n^{(n)}=t_n$.

\begin{lemma}\label{lm:3_6}
\begin{enumerate}
\item
For all $n\in\Z^+$ and $t\in(t_{n-1}, d_0^{(n)})$, $v_y(t)>0$.
\item
For all $n\in\Z^+$, $0\leq i <n$, and $t\in(d_i^{(n)}, e_i^{(n)})$, $v_y(t)<0$.
\item
For all $n\in\Z^+$, $0\leq i<n$, and $t\in (e_i^{(n)},d_{i+1}^{(n)})$, $v_y(t)>0$.
\item
For all $n\in\Z^+$, $0\leq i<n$, and $t\in \{ e_i^{(n)}, d_i^{(n)},t_n\}$, $v_y(t) = 0$.

\item
For all $n\in\Z^+$ and $0\leq i \leq n$, $s_y(d_i^{(n)})= s_y(d_0^{(n)})$.
\item
For all $n\in\Z^+$ and $0\leq i < n$, $s_y(e_i^{(n)}) = s_y(e_0^{(n)})$.
\item
For all $n\in\N$, $s_y(t_n)=\frac{5^3 (e-1)^3}{6^3\cdot 8\pi  } \sum_{i=1}^n \frac{1}{e^{3i}}$.
\item
For all $n\in\Z^+$, $s_y(e_0^{(n)}) = s_y(t_n) -\frac{(e-1)^3}{12^3n^38\pi e^{3n}}$.
\item
$s_y(1) = \frac{5^3 (e-1)^3}{6^3\cdot 8\pi  (e^3-1)}$.
\end{enumerate}
\end{lemma}
\begin{proof}
Parts 1-6 are clear by inspection
and induction. By 4. and Observation \ref{ob:3_4},
\begin{align*}
s_y(t_n) -s_y(t_{n-1})
&=s_y(d_0^{(n)}) - s_y(t_{n-1})\\
&=\frac{[\frac{5}{6}(t_n - t_{n-1}) ]^3}{8\pi} =  \frac{[\frac{5}{6}( (e-1)e^{-n}) ]^3}{8\pi}\\
&=\frac{5^3 (e-1)^3}{6^3\cdot 8\pi e^{3n} }
\end{align*}
for all $n\in\Z^+$, so 6 holds by induction. Also
by 4 and Observation \ref{ob:3_4},
\begin{align*}
s_y(t_n) - s_y(e_0^{(n)})
&=s_y(d_0^{(n)}) - s_y(e_0^{(n)})\\
&=\frac{ [\frac{1}{12n}(t_n-t_{n-1}) ]^3 }{8\pi}=\frac{ [\frac{1}{12n}((e-1)e^{-n}) ]^3 }{8\pi}\\
&=\frac{(e-1)^3}{12^3 n^3 8\pi e^{3n}},
\end{align*}
so 7 holds. Finally, by 6,
\[s_y(1) = \frac{5^3(e-1)^3}{6^3 8 \pi (e^3-1)} ,\]
i.e., 8 holds.
\end{proof}

By Lemmas \ref{lm:3_5} and \ref{lm:3_6}, we see that
$\vec{s}$ parametrizes a curve from $\vec{s}(0)=(0,0)$
to $\vec{s}(1) =  (0, \frac{5^3(e-1)^3}{6^3 8\pi (e^3-1)})$.

The proofs of Lemmas \ref{lm:3_5} and \ref{lm:3_6} are included in the appendix.

It is clear from Observation \ref{ob:3_3} and
Lemmas \ref{lm:3_5} and \ref{lm:3_6} that the curve $\Gamma$ does not intersect itself.
We thus have the following.
\begin{corollary}\label{co:3_7}
${\bf \Gamma}$ is a simple curve
from $\vec{s}(0)=(0,0)$ to $\vec{s}(1) =(0, \frac{5^3(e-1)^3}{6^3 8 \pi (e^3-1)})$.
\end{corollary}
\begin{proof}
Let $\vec{s}':[0,1]\rightarrow \R^2$ be such that
\[\vec{s}'(t)=\begin{cases}
\vec{s}(t_n^+)\frac{t-t_n^-}{t_n^+-t_n^-}+\vec{s}(t_n^-)\frac{t_n^+-t}{t_n^+-t_n^-}&t\in (t_n^-,t_n^+), n\notin K,\\
\vec{s}(t)&\text{otherwise}.
\end{cases}\]
Note that by construction of $\vec{s}$, retracing happens along $y$-axis between
$(0,\vec{s}(t_n^-))$ and $(0, \vec{s}(t_n^+))$ only when $t\in (t_n^-,t_n^+)$ for $n\notin K$.
In $\vec{s}'$, for all $n\notin K$, $\vec{s}'$ maps $(t_n^-,t_n^+)$ to the vertical line segment
between $(0,\vec{s}(t_n^-))$ and $(0, \vec{s}(t_n^+))$ linearly. Otherwise, $\vec{s}'(t) = \vec{s}(t)$.
Hence, $\vec{s}'(0) = (0,0)$, $\vec{s}'(1) = (0,\frac{5^3(e-1)^3}{6^3 8 \pi (e^3-1)})$,
and $\vec{s}'$ is a one-to-one parametrization of $\Gamma=\range(\vec s)$, although $\vec{s}'$ is not
computable. Therefore $\Gamma$ is a simple curve.
\end{proof}
\begin{lemma}\label{lm:3_8}
The functions $\vec{a}, \vec{v}$, and $\vec{s}$ are
Lipschitz, hence continuous, on $[0,1]$.
\end{lemma}
\begin{proof}
It is clear by differentiation that
$Lip(\varphi_{a,b})=\frac{\pi}{2}$ for all $a,b\in\R$ with $a<b$.
It follows by inspection that $Lip(a_x)\leq\frac{\pi}{4}$
and $Lip(a_y)=\frac{\pi}{2}$, whence
\[Lip(\vec{a})\leq \sqrt{Lip(a_x)^2 +Lip(a_y)^2}\leq \frac{\pi\sqrt{5}}{4}.\]
Thus $\vec{a}$ is Lipschitz, hence continuous (and
bounded), on $[0, 1]$. It follows immediately
that $\vec{v}$ and $\vec{s}$ are Lipschitz, hence continuous,
on $[0,1]$.
\end{proof}

Since every Lipschitz parametrization has
finite total length \cite{Apos76a}, and since the
length of a curve cannot exceed the total
length of any of its parametrizations, we
immediately have the following.
\begin{corollary}\label{co:3_9}
The total length, including
retracings, of the parametrization $\vec{s}$ is
finite. Hence the curve ${\bf \Gamma}$ is rectifiable.
\end{corollary}
\begin{lemma}\label{lm:3_10}
The curve ${\bf \Gamma}$ is smooth
except at the endpoint $\vec{s}(1)$.
\end{lemma}
\begin{proof}
We have seen that ${\bf \Gamma}([0,t_1^-])$ is
simply a segment of the $y$-axis, and that
the vector velocity function $\vec{v}$ is continuous
on $[0,1]$. Since the set
\[Z=\myset{t\in(0,1) }{\vec{v}(t)=0}\]
has no accumulation points in $(0,1)$, it
therefore suffices to verify that, for each
$t^*\in Z$,
\begin{equation}\label{eq:3_1}
\lim_{t\rightarrow t^{*-}}\frac{\vec{v}(t) }{\card{\vec{v}(t) }}
=\lim_{t\rightarrow t^{*+}}\frac{\vec{v}(t) }{\card{\vec{v}(t) }},
\end{equation}
i.e., that the left and right tangents of
${\bf \Gamma}$ coincide at $\vec{s}(t^*)$. But this is
clear, because Lemmas \ref{lm:3_5} and \ref{lm:3_6}
tell us that
\[Z=\myset{t_n }{ n\in\Z^+\text{ and } \tau(n)=\infty},\]
and both sides of \eqref{eq:3_1} are $(0,1)$ at
all $t^*$ in this set.
\end{proof}

%
%
%


\begin{lemma}\label{lm:3_11}
The functions $\vec{a}, \vec{v}$, and $\vec{s}$
are computable in polynomial time. The total length including retracings, of $\vec{s}$ is computable in polynomial time.
\end{lemma}
\begin{proof}This follows from Observation \ref{ob:3_4}, Lemmas \ref{lm:3_5} and \ref{lm:3_6}, and the polynomial time computability of $f(n)=\sum_{i=1}^n e^{-3i}$.
\end{proof}

\begin{definition}
A {\em modulus of uniform continuity} for a function $f:[a,b]\rightarrow\R^n$ is a
function $h:\N\times\N$ such that, for all $s,t\in[a,b]$ and $r\in\N$,
\[|s-t|\leq 2^{-h(r)}\implies |f(s)-f(t)|\leq 2^{-r}.\]
\end{definition}
It is well known (e.g., see \cite{Ko91})
that every computable function $f:[a,b]\rightarrow\R^n$
has a modulus of uniform continuity that is
continuous.
\begin{lemma}\label{lm:3_12}
Let $f:[a,b]\rightarrow\R^2$ be a parametrization of $\bf\Gamma$. If $f$ has bounded retracing
and a computable modulus of uniform continuity, then $\K\leq_\mathrm{T} f_y$, where $f_y$ is the vertical
component of $f$.
\end{lemma}
\begin{proof}
Assume the hypothesis. Then there exist $m\in\Z^+$ and $h:\N\rightarrow\N$ such that
$f$ does not have $m$-fold retracing and $h$ is a computable modulus of uniform
continuity for $f$. Note that $h$ is also a modulus of uniform continuity for $f_y$.

Let $M$ be an oracle Turing machine that, given an oracle $\mathcal{O}_g$ for a function $g:[a,b]\rightarrow\R$,
implements the algorithm in Figure \ref{fig:3_w}. The
key properties of this algorithm's choice of $r$ and $\Delta$ are that the following hold when $g=f_y$.
\begin{enumerate}[(i)]
\item
For each time $t$ with $f_y(t)=s_y(t_n)$, there
is a nearby time $\tau_j$ with $j$ high. Similarly for $f_y(t)=s_y(e_0^{(n)})$ and $j$ low.
\item
For each high $j$, $|f_y(\tau_j)-s_y(t_n)|\leq 3\cdot 2^{-r}$.
Similarly for each low $j$ and $s_y(e_0^{(n)})$.
\item
No $j$ can be both high and low.
\end{enumerate}
Now let $n\in\Z^+$. We show that $M^{\mathcal{O}_{f_y}}(n)$
accepts if $n\in\K$ and rejects if $n\notin\K$.
This is clear if $n\leq m$, so assume that $n>m$.
\begin{figure}[htbp]
\begin{center}
\begin{tabbing}
abc \=abc\=abc\=abc\=abc\=abc\=abc \kill\\
\>       {\bf input} $n \in\Z^+$;\\
\>       {\bf if} $n\leq m$ {\bf then}\\
\>       use a finite lookup table to accept if $n\in\K$ and reject if $n\notin\K$\\
\>       {\bf else}\\
\>       {\bf begin}\\
\>\>         $r$:= the least positive integer such that $2^{3-r}< s_y(t_n)-s_y(e_0^{(n)})$;\\
\>\>         $\Delta$:=$2^{-h(r)}$;\\
\>\>         {\bf for} $0\leq j\leq (b-a)/\Delta$ {\bf do}\\
\>\>         {\bf begin}\\
\>\>\>            $\tau_j$:=$a+\Delta_j$;\\
\>\>\>            call $j$ {\bf high} if $|\mathcal{O}_g(\tau_j, r)-s_y(t_n)|<2^{1-r}$\\
\>\>\>            call $j$ {\bf low} if $|\mathcal{O}_g(\tau_j,r)-s_y(e_0^{(n)}| < 2^{1-r}$\\
\>\>         {\bf end};\\
\>\>         {\bf if} there is a sequence $0<j_0<j_1<\cdots <j_m$ in which $j_i$ is high for all even $i$ and low for all odd $i$\\
\>\>         {\bf then} accept\\
\>\>         {\bf else} reject\\
\>     {\bf end}.\\
\end{tabbing}
\caption{Algorithm for $M^{\mathcal{O}_g}(n)$ in the proof of Lemma \ref{lm:3_12}.}
\label{fig:3_w}
\end{center}
\end{figure}

If $n\in\K$, then Observation \ref{ob:3_3}, Lemma \ref{lm:3_5}, and Lemma \ref{lm:3_6}
tell us that $M^{\mathcal{O}_{f_y}}(n)$ accepts.
If $n\notin\K$, then the fact that $f$ does not have $m$-fold retracing tells us that
$M^{\mathcal{O}_{f_y}(n)}$ rejects.
\end{proof}
\begin{proof}[Proof of Theorem \ref{th:3_2}]
Part 1 follows from Lemmas \ref{lm:3_8} and \ref{lm:3_11}. Part 2 follows
from Lemma \ref{lm:3_11}. Part 3 follows from Corollaries \ref{co:3_7} and \ref{co:3_9} and Lemma \ref{lm:3_10}.
Part 4 follows from Lemma \ref{lm:3_12}, the
fact that every computable function $g:[a,b]\rightarrow\R^2$
has a computable modulus of uniform continuity,
and the fact that $A$ is decidable wherever $A\leq_{\mathrm{T}}g$ and $g$ is computable.
\end{proof}
\end{section} 
\begin{section}{Lower Semicomputability of Length}\label{se:4}
In this section we prove that every computable curve $\Gamma$ has a lower semicomputable length. Our proof is somewhat involved, because our result holds even if every computable parametrization of $\Gamma$ is retracing.

\begin{construction}\label{con:4_1}
Let $f:[0,1]\rightarrow \R^n$ be a computable function.
Given an oracle Turing machine $M$ that computes $f$
and a computable modulus $m:\N\rightarrow \N$ of the
uniform continuity of $f$, the $(M,m)$-cautious
polygonal approximator
of $\range(f)$ is the function
    $\pi_{M,m}:\N\rightarrow \{polygonal\; paths\}$
computed by the following algorithm.

\begin{tabbing}
abc \=abc\=abc\=abc\=abc\=abc\=abc \kill\\
\>       {\bf input} $r \in\N$;\\
\>       $S:=\{\}$; // $S$ may be a multi-set\\
\>       {\bf for} $i$:=$0$ {\bf to} $2^{m(r)}$ {\bf do}\\
\>\>         $a_i:=i2^{-m(r)}$;\\
\>\>         use $M$ to compute $x_i$ with\\
\>\>\>         $|x_i-f(a_i)|\leq 2^{-(r+m(r)+1)} $;\\
\>\> add $x_i$ to S;\\
\> output a longest path inside a minimum spanning tree of $S$.\\
\end{tabbing}
\end{construction}

\begin{defn}
Let $(X,d)$ be a  metric space.
Let  $\Gamma\subseteq X$ and $\epsilon>0$.
Let
\[\Gamma(\epsilon) =\myset{p\in X}{ \inf_{p' \in \Gamma} d(p,p')\leq \epsilon }\]
be the {\em Minkowski sausage} of $\Gamma$ with radius $\epsilon$.

Let $\hd:\mathcal{P}(X)\times\mathcal{P}(X)\rightarrow\R$ be
such that for all $\Gamma_1, \Gamma_2\in \mathcal{P}(X)$
\[\hd(\Gamma_1,\Gamma_2) =
\inf\myset{\epsilon }{ \Gamma_1\subseteq \Gamma_2(\epsilon)\text{ and }
\Gamma_2\subseteq \Gamma_1(\epsilon)}.\]
\end{defn}
Note that $\hd$ is the {\em Hausdorff distance} function.

Let $\mathcal{K}(X)$ be the set of nonempty compact subsets of  $X$.
Then $(\mathcal{K}(X), \hd)$ is a metric space \cite{Edgar90}.

\begin{theorem}\label{th:compactspace}(Frink \cite{Frink:TL}, Michael \cite{Michael:TSS}).
Let $(X, d)$ be a compact metric space.
Then $(\mathcal{K}(X), \hd)$ is a compact metric space.
\end{theorem}

\begin{defn}
Let $\mathcal{RC}$ be the set of all simple rectifiable curves in $\R^n$.
\end{defn}

\begin{theorem}\label{th:liminflength}(\cite{Tricot95} page 55).
Let $\Gamma\in\mathcal{RC}$.
Let $\{ \Gamma_n\}_{n\in\N}\subseteq \mathcal{RC}$ be a sequence of rectifiable curves
such that $\limn \hd(\Gamma_n, \Gamma) =0$.
Then $\CH^1(\Gamma) \leq \liminfn \CH^1(\Gamma_n)$.
\end{theorem}

This theorem has the following consequence.
\begin{theorem}\label{th:4_4}
Let $\Gamma\in\mathcal{RC}$.
For all $\epsilon>0$, there exists $\delta>0$ such that
for all $\Gamma'\in\mathcal{RC}$, if $\hd(\Gamma, \Gamma')<\delta$,
then $\CH^1(\Gamma')> \CH^1(\Gamma)-\epsilon$.
\end{theorem}

In the following, we prove a few technical lemmas that lead to Lemma \ref{lm:extendedlengthequality},
which plays an important role in proving Theorem \ref{th:celength}.

\begin{lemma}\label{lm:propersubset}
Let $\Gamma\in\mathcal{RC}$.
Let $p_0, p_1,\in\Gamma$ be its two endpoints.
Let $\Gamma'\subsetneq \Gamma$ such that $p_0, p_1\in \Gamma'$.
Then $\Gamma'\notin \mathcal{RC}$.
\end{lemma}
\begin{proof}
If $\Gamma'$ is not closed, then we are done.
Assume that $\Gamma'$ is closed.
Let $\gamma$ be a parametrization of $\Gamma$ such that
$\gamma(0) = p_0$ and $\gamma(1)=p_1$.

Since $\Gamma'\neq \Gamma$ and  $p_0, p_1\in \Gamma'$,
$\gamma^{-1}(\Gamma') \subseteq I_0\cup I_1$,
where $I_0\subseteq [0,1]$ and $I_1\subseteq [0,1]$ are
closed and disjoint.

It is easy to see that $\gamma(I_0)$ and $\gamma(I_1)$ are
closed and disjoint.
And thus, for any continuous function $\gamma':[0,1]\rightarrow \R^n$,
$\gamma'^{-1}(\gamma(I_0))$ and $\gamma'^{-1}(\gamma(I_1))$
are closed and disjoint.
Therefore, for any continuous function $\gamma':[0,1]\rightarrow \R^n$,
$\gamma^{-1}(\Gamma')\neq [0,1]$, i.e., $\Gamma'\notin \mathcal{RC}$.
\end{proof}

\begin{lemma}\label{lm:connected}
Let $\Gamma\in\mathcal{RC}$.
Let $\Gamma'\subseteq\Gamma$ be a connected compact set.
Then $\Gamma'\in \mathcal{RC}$.
\end{lemma}
\begin{proof}
Let $\gamma$ be the parametrization of $\Gamma$.

Let $a =\inf \{ \gamma^{-1}(\Gamma')\}$
and let $b = \sup \{ \gamma^{-1}(\Gamma')\}$.

Let $\gamma':[0,1]\rightarrow \R^n$ be such that for all $t\in[0,1]$
\[\gamma'(t) = \gamma(a+t(b-a)).\]
Then $\gamma'$ defines a curve and
we show that $\gamma'([0,1])=\Gamma'$.

It is clear that $\Gamma'\subseteq \gamma'([0,1])$.
Since $\Gamma'$ is compact, we know that
$\gamma'(0), \gamma'(1)\in\Gamma'$.

Suppose for some $t'\in(0,1)$, $\gamma'(t')\notin \Gamma'$.
Since $\Gamma'$ is compact, there exists $\epsilon>0$ such that
$\gamma'([ t'-\epsilon , t'+\epsilon]) \cap \Gamma' =\varnothing$.
Then
$\Gamma'\subseteq \gamma'([0, t'-\epsilon))\cup \gamma'((t'+\epsilon, 1])$.
Since $\gamma'$ is one-one,
\[\hd(\gamma'([0, t'-\epsilon)),  \gamma'((t'+\epsilon, 1])  ) >0.\]
Hence,
\[\hd(\Gamma'\cap \gamma'([0, t'-\epsilon)),  \Gamma'\cap \gamma'((t'+\epsilon, 1])  ) >0.\]
Thus, $\Gamma'$ cannot be connected.

Therefore, if $\Gamma'$ is connected,
then $\Gamma'= \gamma'([0,1])$ and hence $\Gamma'\in\mathcal{RC}$.
\end{proof}

\begin{lemma}\label{lm:connectedlimit}
Let $\Gamma_0, \Gamma_1,\dots$ be a convergent sequence of compact sets in compact metric
space $(X,d)$ that is eventually connected.
Let $\Gamma = \limn \Gamma_n$.
Then $\Gamma$ is connected.
\end{lemma}
\begin{proof}
We prove the contrapositive.

Assume that $\Gamma$ is not connected.
Then there exists open sets $A, B\subseteq X$ such that
$A\cap B =\varnothing$, $\Gamma\cap A\neq \varnothing$,
$\Gamma\cap B\neq \varnothing$, and $\Gamma\subseteq A\cup B$.

Then $(\Gamma\cap A) \cap (\Gamma\cap B) =\varnothing$,
thus $\hd(\Gamma\cap A, \Gamma\cap B) >0$.
Let
\[\delta =\hd(\Gamma\cap A, \Gamma\cap B).\]

Since $\limn\Gamma_n =\Gamma$, let $n_0$ be such that for all $n\geq n_0$,
\[\hd( \Gamma_n, \Gamma) \leq \tfrac{\delta}{3}.\]

It is clear that
\[ (\Gamma\cap A)(\tfrac{\delta}{3}) \cap  \Gamma_n  \neq \varnothing,\]
\[ (\Gamma\cap B)(\tfrac{\delta}{3}) \cap  \Gamma_n  \neq \varnothing,\]
and
\[\Gamma_n \subseteq  (\Gamma\cap A)(\tfrac{\delta}{3})  \cup  (\Gamma\cap B)(\tfrac{\delta}{3}).\]
By the definition of $\delta$,
\[\hd((\Gamma\cap A)(\tfrac{\delta}{3}),  (\Gamma\cap B)(\tfrac{\delta}{3})) \geq \tfrac{\delta}{3}.\]
Thus $\Gamma_n$ is not connected for all $n\geq n_0$.
\end{proof}

\begin{lemma}\label{lm:lengthequality}
Let $\Gamma\in \mathcal{RC}$ and let
$f:[0,1]\rightarrow\Gamma$ be a parametrization of $\Gamma$.
Let
\[ L(\Gamma, \epsilon) =\inf\myset{\CH^1(\Gamma')}{\Gamma' \in \mathcal{RC}
\text{ and }\Gamma'\subseteq \Gamma(\epsilon)
\text{ and } f(0), f(1) \in\Gamma'  }.\]
Then
\[\lim_{\epsilon\rightarrow 0^+} L(\Gamma, \epsilon) = \CH^1(\Gamma).\]
\end{lemma}
\begin{proof}
It is clear that
$\lim_{\epsilon\rightarrow 0^+} L(\Gamma, \epsilon) \leq \CH^1(\Gamma)$.
It suffices to show that
$\lim_{\epsilon\rightarrow 0^+} L(\Gamma, \epsilon) \geq \CH^1(\Gamma)$.

Let $\delta>0$.
For each $i\in\N$,
let
\[ S_i =\myset{\Gamma' \in \mathcal{RC} }{
\Gamma' \subseteq \Gamma(\tfrac{1}{i})
\text{ and }
\gamma(0), \gamma(1)\in \Gamma'
},\]
where $\gamma$ is a
parametrization of $\Gamma$.
Note that if $i_2 <i_1$,
then $S_{i_1}\subseteq S_{i_2}$.

Let
$\Gamma_0, \Gamma_1,\dots $ be an arbitrary sequence such that
for all $i\in\N$, $\Gamma_i\in S_{k_i}$,
and $k_0, k_1, \dots \in\N$ is a strictly increasing sequence.

Since for all $i\in\N$, $\Gamma_i$ is compact and connected,
by Theorem \ref{th:compactspace} and Lemma \ref{lm:connectedlimit},
there is at least one  cluster point and every cluster point
is a connected compact set.
Let $\Gamma'$ be a cluster point.
It is clear that $\Gamma'\subseteq \Gamma$.
Then by Lemma \ref{lm:connected},
$\Gamma'\in\mathcal{RC}$.

It is also clear that $\gamma(0), \gamma(1)\in \Gamma'$
by definition of $S_i$.
Thus by Lemma \ref{lm:propersubset},
$\Gamma' =\Gamma$.

By Theorem \ref{th:liminflength},
$\liminfn \CH^1(\Gamma_n)\geq \CH^1(\Gamma') =\CH^1(\Gamma)$.
Then by Theorem \ref{th:4_4}, this implies that  for all sufficiently large $i\in\N$,
\[(\forall \Gamma''\in S_i) \CH^1(\Gamma'') \geq \CH^1(\Gamma)-\delta.\]
Therefore, for all sufficiently large $i\in\N$,
$L(\Gamma,\tfrac{1}{i}) \geq \CH^1(\Gamma)-\delta$.
Since $\delta>0$ is arbitrary,
\[\lim_{\epsilon\rightarrow 0^+} L(\Gamma, \epsilon)\geq \CH^1(\Gamma).\]
\end{proof}

\begin{lemma}\label{lm:extendedlengthequality}
Let $\Gamma\in \mathcal{RC}$ and let
$f:[0,1]\rightarrow\Gamma$ be a parametrization of $\Gamma$.
Let
\[ L(\Gamma, \epsilon, p_1, p_2) =\inf\myset{\CH^1(\Gamma')}{\Gamma' \in \mathcal{RC}
\text{ and }\Gamma'\subseteq \Gamma(\epsilon)
\text{ and } p_1, p_2 \in\Gamma'  }.\]
Then
\[\lim_{\epsilon\rightarrow 0^+} \sup_{p_1, p_2\in \Gamma(\epsilon) }L(\Gamma, \epsilon,p_1,p_2) = \CH^1(\Gamma).\]
\end{lemma}
\begin{proof}
For every $p\in \Gamma(\epsilon)$, there exists a point $p'\in\Gamma$ such
that $\lVert p, p'\rVert \leq \epsilon$ and line segment $[p, p'] \subseteq \Gamma(\epsilon)$.
Thus it is clear that for all $p_1, p_2 \in\Gamma(\epsilon)$,
$L(\Gamma,\epsilon, p_1,p_2)\leq 2\epsilon +\CH^1(\Gamma)$.
Therefore,
\[\lim_{\epsilon\rightarrow 0^+} \sup_{p_1, p_2\in \Gamma(\epsilon) }L(\Gamma, \epsilon,p_1,p_2) \leq \CH^1(\Gamma).\]

For the other direction, observe that
\[\lim_{\epsilon\rightarrow 0^+} \sup_{p_1, p_2\in \Gamma(\epsilon) }L(\Gamma, \epsilon,p_1,p_2)
\geq \lim_{\epsilon\rightarrow 0^+}L(\Gamma, \epsilon).\]
Applying Lemma \ref{lm:lengthequality} completes the proof.
\end{proof}

\newcommand{\mesh}{\mathrm{mesh}}

\begin{theorem}\label{th:celength}
Let $\Gamma\in \mathcal{RC}$ such that $\Gamma = \gamma([0,1])$,
where $\gamma$ is a continuous function.
(Note that $\gamma$ may not be one-one.)
Let $S(a)=\myset{\gamma(a_i)}{ a_i\in a }$
for all dissection $a$.
Let $\{ a_n\}_{n\in\N}$ be a sequence of
dissections of $\Gamma$ such that
\[\limn \mesh(a_n) =0.\]
Then
\[\limn \CH^1(LMST(a_n)) =\CH^1(\Gamma),\]
where $LMST(a)$ is the longest path inside the Minimum Euclidean Spanning Tree of
$S(a)$.
\end{theorem}
\begin{proof}
For all $n\in\N$, let
\[\epsilon_n = 2 \hd(\Gamma, S(a_n)).\]
Note that since $\gamma$ is uniformly continuous
and $\limn \mesh(a_n) =0$,
$\limn \epsilon_n =0$.

Let $w = 2\epsilon_n$.

\begin{claim}
Let $T$ be a Euclidean Spanning Tree of
$S(a)$. If $T$ has an edge that is not inside $\Gamma(w)$,
then $T$ is not a minimum spanning tree.
\end{claim}
\begin{proof}[Proof of Claim]
Let $E$ be an edge of $T$ such that $E\nsubseteq \Gamma(w)$.
Then $\CH^1(E)>2w$.
Removing $E$ from $T$ will break $T$ into two subtrees $T_1$, $T_2$.
By the definition of $\epsilon_n$ and the continuity of $\gamma$,
there exists $s_1, s_2\in S(a)$ with $\lVert s_1 -s_2\rVert\leq \epsilon_n$
such that $s_1 \in T_1$ and $s_2\in T_2$.

It is clear that $T_1\cup T_2 \cup \{ (s_1, s_2)\}$ is also
a Euclidean Spanning Tree of $S(a)$ and
$\CH^1(T_1\cup T_2 \cup \{ (s_1, s_2)\}) < \CH^1(T)$,
i.e., $T$ is not minimum.
\end{proof}

Let $T$ be a Minimum Euclidean Spanning Tree of $S(a)$.
Let $L$ be the longest path inside $T$.
Then $L\subseteq T \subseteq \Gamma(w)$.

Note that $\CH^1(L)\leq \CH^1(\Gamma)$.

Let $p_0, p_1$ be the two endpoints of $\Gamma$.

Since $L$ is the longest path inside $T$
and $p_0$, $p_1$ are each within $\epsilon_n$ distance
to some point in $S(a_n)$,
\[L(\Gamma, w, p_0,p_1)\leq 2\epsilon_n+\CH^1(L).\]

By Lemma \ref{lm:extendedlengthequality},
\[\lim_{w\rightarrow 0^+} L(\Gamma, w,p_0,p_1) =\CH^1(\Gamma).\]
Then
\[\limn \CH^1(LMST(a_n)) =\CH^1(\Gamma).\]
\end{proof}
This result implies that when the sampling
density is high, the number of leaves in the minimum spanning
tree is asymptotically smaller than the total number of
nodes.

We now have the machinery to prove the main result of this section.

\begin{theorem}\label{th:celengthmain}
Let $\gamma:[0,1]\rightarrow \R^n$ be computable such that $\Gamma=\gamma([0,1])\in\mathcal{RC}$.
Then $\CH^1(\Gamma)$ is lower semicomputable.
\end{theorem}
\begin{proof}
Let the function $f$, $M$, and $m$ in Construction \ref{con:4_1} be $\gamma$,
a computation of $\gamma$, and its computable modulus respectively.

For each input $r\in\N$, $\pi_{M,m}(r)$ is the longest path $L_r$ in
$MST(S_r)$, where $S_r$ is the set of points sampled by $\pi_{M,m}(r)$.

Let $l_r = \CH^1(L_r) -2^{-r}$.
Note that $l_r$ is computable from $r\in\N$.

We show that for all $r\in\N$, $l_r \leq \CH^1(\Gamma)$ and
$\lim_{r\rightarrow\infty} l_r = \CH^1(\Gamma)$.

Let $\tilde{f}$ be a one-one parametrization of $\Gamma$.
Let $\pi:\{0, \dots, 2^{m(r)}\}\rightarrow \{0, \dots, 2^{m(r)}\}$
be a permutation of $\{0, \dots, 2^{m(r)}\}$ such that
for all $i, j \in \{0, \dots, 2^{m(r)}\}$,
\[i < j\implies \tilde{f}^{-1}( f(a_{\pi(i)}) ) < \tilde{f}^{-1}( f(a_{\pi(j)}) ).\]

Let $\hat{\Gamma}_r$ be
the polygonal curve connecting
the points $f(a_{\pi(0)}), f(a_{\pi(1)}),\dots, f(a_{\pi(2^{m(r)})})$  in order.
Then $\hat{\Gamma}_r$ is a polygonal approximation of $\Gamma$
and $\CH^1(\hat\Gamma_r) \leq \CH^1(\Gamma)$.

Let $\bar{\Gamma}_r$ be the polygonal curve connecting
the points in $S_r$  in the order of $x_{\pi(0)}, x_{\pi(1)}, \dots, x_{\pi(2^{m(r)})}$.

Due to the approximation induced by the computation in Construction \ref{con:4_1},
\[\CH^1(\bar{\Gamma}_r) \leq \CH^1(\hat{\Gamma}_r) + 2^{-r}.\]
Then it is clear that
\[\CH^1(L_r )=\CH^1( LMST(S_r) )\leq\CH^1(\bar{\Gamma}_r)   \leq \CH^1(\hat{\Gamma}_r) + 2^{-r}.\]
Thus
\[l_r \leq \CH^1(\hat{\Gamma}_r) .\]

Let $\hat{S}_r =\{f(a_0), f(a_1), \dots, f(a_{2^{m(r)}})\}$. Note that $\hat{S}_r$ may be
a multi-set.
By Theorem \ref{th:celength},
\[\lim_{r\rightarrow \infty} LMST(\hat{S}_r) =\CH^1(\Gamma).\]

Let
\[\epsilon_r =2\hd(\Gamma, S_r).\]
By Contruction \ref{con:4_1},
\[\lim_{r\rightarrow \infty} \epsilon_r =0.\]
Let $w_r=2\epsilon_r$.

Let $T_r$ be a Minimum Euclidean Spanning Tree of $S_r$.
Let $L_r$ be the longest path inside $T_r$.
By the Claim in Theorem \ref{th:celength},
$L\subseteq T\subseteq \Gamma(w_r)$.

By an essentially identical argument as the one in the proof of Theorem \ref{th:celength},
\[\lim_{r\rightarrow\infty } l_r =\lim_{r\rightarrow\infty } \CH^1(LMST(S_r)) =\CH^1(\Gamma),\]
which completes the proof.
\end{proof}

\end{section} 
\begin{section}{$\Delta_2^0$-Computability of the Constant-Speed Parametrization}\label{se:5}

In this section we prove that every computable curve $\Gamma$ has a constant speed parametrization that is
$\Delta_2^0$-computable.

\begin{theorem}\label{th:5_1}
Let $\Gamma=\gamma^*([0,1])\in\mathcal{RC}$. ($\gamma^*$ may not be one-one.)
Let $l=\CH^1(\Gamma)$ and $O_l$ be an oracle such that
for all $n\in\N$, $\lvert O_l(n) - l\rvert \leq 2^{-n}$.
Let $f$ be a computation of $\gamma^*$ with modulus $m$.
Let $\gamma$ be the constant speed parametrization of $\Gamma$.
Then $\gamma$ is computable with oracle $O_l$.
\end{theorem}
\begin{proof}
On input $k$ as the precision parameter for computation of the curve
and a rational number $x\in [0,1]\cap \Q$, we output a point $f_k(x) \in \R^n$
such that
$|f_k(x) - \gamma(x) |\leq 2^{-k}$.

Without loss of generality, assume that $\CH^1(\Gamma)> 1000 \cdot 2^{-k}$.

Let $\delta = 2^{-(4+k)}$.

Run $f$ as in Construction \ref{con:4_1} with increasingly larger precision parameter $r> - \log \delta$
until
\[\CH^1(LMST(a))> \CH^1(\Gamma)-\tfrac{\delta}{2}\]
and the shortest distance between the two endpoints of
$LMST(a)$ inside the polygonal sausage around $LMST(a)$ with
width $2d=2\cdot 2^{-r}$ is at least $\CH^1(\Gamma) -\tfrac{\delta}{2}$.
This can be achieved by using Euclidean shortest path algorithms \cite{Kapoor:EAESPVPPO,Hershberger:OAESPP}.

Let $d_k\leq  2^{-(4+k)}$ be the largest $d$ such that the above conditions are
satisfied, which is assured by Theorem \ref{th:celengthmain} and Lemma \ref{lm:extendedlengthequality}.
Let $\mathcal{S}$ be the polygonal sausage around $LMST(a)$ with
width $2d_k$.

For $p_1, p_2 \in \mathcal{S}$,
let
$d_{\mathcal{S}}(p_1, p_2)=\text{the shortest
distance between $p_1$ and $p_2$ inside $\mathcal{S}$}$.
Note that $\mathcal{S}$ is connected.

Let $f_k$ be the constant speed parametrization of $LMST(a)$
and $\gamma$ be the constant speed parametrization of $\Gamma$.
Without loss of generality, assume that
$\lVert \gamma(0)-f_k(0)\rVert < \lVert \gamma(1)-f_k(0)\rVert$
and
$\lVert \gamma(1)-f_k(1)\rVert < \lVert \gamma(0)-f_k(1)\rVert$,
since we can hardcode approximate locations of $\gamma(0)$ and $\gamma(1)$
such that when $d_k$ is sufficiently small, we can decide
wehther a sampled point is closer to $\gamma(0)$ or $\gamma(1)$.
As we now prove
\[\lim_{k\rightarrow\infty} \{f_k(0), f_k(1)\}  = \{\gamma(0), \gamma(1)\}.\]

Note that for each $s\in S$ such that $s\notin LMST(a)$, there exists
$p\in LMST(a)\cap S$ such that
the shortest path from $s$ to $p$ in $MST(a)$
has length less than $\tfrac{\delta}{2}$, i.e.,
$d_{MST(a)}(s, p) < \tfrac{\delta}{2}$,
since $\CH^1(LMST(a)) > \CH^1(\Gamma) -\tfrac{\delta}{2}$
and $\CH^1(MST(a)) \leq \CH^1(\Gamma)$.

Let
$\delta_0 =d_{\mathcal{S}}( \gamma(0), f_k(0) )$.
Let $s_0$ be the closest point to $\gamma(0)$ in $S\cap LMST(a)$.
Then
$d_{\mathcal{S}} ( \gamma(0),  s_0) \leq \tfrac{\delta}{2}+d_k$.
Then
$d_{LMST(a)}( s_0,  f_k(0)) \geq\delta_0 -\tfrac{\delta}{2}-d_k$.
Since $s_0\in S\cap LMST(a)$ and we assume $\CH^1(\Gamma)> 1000 \cdot 2^{-k}$,
\[d_{\mathcal{S}}( s_0,\gamma(1))
\leq  \CH^1(LMST(a)) -\delta_0+\tfrac{\delta}{2} + d_k + \tfrac{\delta}{2} +d_k
=\CH^1(LMST(a)) -\delta_0+\delta + 2d_k .\]
Then
\begin{align*}
d_{\mathcal{S}}( \gamma(0),\gamma(1))
&\leq \CH^1(LMST(a)) -\delta_0+\delta + 2d_k  + \tfrac{\delta}{2}+d_k\\
& <\CH^1(LMST(a)) -\delta_0 +\tfrac{3\delta}{2} + 3d_k.\\
\end{align*}
And hence
\begin{equation}\label{eq:a4_1}
d_{\mathcal{S}}( \gamma(0),\gamma(1))
\leq \CH^1(\Gamma)  -\delta_0 +2\delta + 3d_k.
\end{equation}

By the choice of $d_k$,
we have that
$ d_{\mathcal{S}}( f_k(0), f_k(1)) \geq \CH^1(\Gamma) -\tfrac{\delta}{2}$.
Now, note that for any two points $p_1, p_2\in \Gamma$,
\[d_{\mathcal{S}}(p_1, p_2)\leq  \frac{\CH^1(\Gamma) + d_{\mathcal{S}}(\gamma(0), \gamma(1))}{2},\]
since we can put them in half of a loop.
Therefore
\[ d_{\mathcal{S}}( f_k(0), f_k(1)) \leq \frac{\CH^1(\Gamma) + d_{\mathcal{S}}(\gamma(0), \gamma(1))}{2}.\]
Thus
\begin{equation}\label{eq:a4_2}
d_{\mathcal{S}}( \gamma(0), \gamma(1)) \geq \CH^1(\Gamma) -\delta .
\end{equation}
By \eqref{eq:a4_1} and \eqref{eq:a4_2}, we have
\begin{equation}
\delta_0\leq 3\delta+3d_k \leq 6\delta < 2^{-k},
\end{equation}
i.e.,
\begin{equation}
\lVert f_k(0) - \gamma(0)  \rVert\leq d_{\mathcal{S}}(f_k(0), \gamma(0)) \leq 6\delta  < 2^{-k}.
\end{equation}
Similarly,
\begin{equation}
\lVert f_k(1) - \gamma(1)  \rVert \leq d_{\mathcal{S}}(f_k(1), \gamma(1)) \leq 6\delta< 2^{-k}.
\end{equation}

Now we proceed to show that for all $t\in (0,1)$,
$\lVert f_k(t) - \gamma(t) \rVert < 10\delta$
with $f(0)$ being at most $6\delta$ from
$\gamma(0)$ inside $\mathcal{S}$ and   $f(1)$ being at most $6\delta$ from
$\gamma(1)$ inside $\mathcal{S}$.

Let $\Delta_k = \lVert f_k(t) -\gamma(t) \rVert$.

Let $s_f \in S\cap LMST(a)$ be such that $| f_k^{-1}(s_f)  -  t |$
is minimized.
Then
$d_{LMST(a)}(f_k(t) ,  s_f)\leq d_k$,
since every edge in $MST(a)$ is at most $d_k$ long.

Let $s'_\gamma \in S\cap  \Gamma$ be such that $ |\gamma^{-1}(s'_\gamma) - t  | $
is minimized. Then
$d_{\Gamma}(\gamma(t), s'_\gamma)\leq d_k$,
since we sample $S$ using $d_k$ as the density parameter.

Let $s_\gamma\in S\cap LMST(a)$
such that $d_{MST(a)}(s_\gamma, s'_\gamma)$ is minimized.
Then
$d_{MST(a)}(s_\gamma, s'_\gamma)\leq \tfrac{\delta}{2}$,
since $\CH^1(MST(a)) \geq \CH^1(\Gamma)-\tfrac{\delta}{2}$.

Then
$\lVert f_k(t) - s_\gamma\rVert\geq \Delta_k  - (\tfrac{\delta}{2}+d_k) =\Delta_k-\tfrac{\delta}{2}-d_k$.

Note that
$d_{LMST(a)}(s_f , s_\gamma)\geq \lVert s_f - s_\gamma\rVert\geq \Delta_k-\tfrac{\delta}{2}-2d_k$.

Without loss of generality, assume that
distance from $s_\gamma$ to $f_k(0)$ along $LMST(a)$
is $\Delta_k-\tfrac{\delta}{2}-d_k$ more than the distance
from $f_k(t)$ to $f_k(0)$. Otherwise, we simply
look from the $\gamma(1)$ and $f_k(1)$ side instead.

The path traced by $\gamma$ from $\gamma(0)$ to $\gamma(t)$ has length $t \cdot \CH^1(\Gamma)$.

The shortest distance between $\gamma(t)$ to $s_\gamma$
inside $\Gamma \cup MST(a)$ is at most $d_k+\tfrac{\delta}{2}$.

The path traced by $f_k$ from $s_\gamma$ to $f_k(1)$ has length
\begin{align*}
d_{LMST(a)}(s_\gamma, f_k(1) )
&\leq
\CH^1(LMST(a))
-[t(\CH^1(\Gamma) -\tfrac{\delta}{2} ) -d_k +\Delta_k-\tfrac{\delta}{2}-d_k].
\end{align*}

The shortest distance from $\gamma(1)$ to $f_k(1)$ inside
$\mathcal{S}$ is at most $6\delta $.

Then the distance from $\gamma(0)$ to $\gamma(1)$
inside $\mathcal{S}$ is at most
\begin{align*}
&t\cdot \CH^1(\Gamma)
+ d_k+\tfrac{\delta}{2}
+ \CH^1(LMST(a)) -[t(\CH^1(\Gamma)-\tfrac{\delta}{2}) -d_k+\Delta_k-\tfrac{\delta}{2}-d_k]+6\delta\\
&\leq \CH^1(LMST(a)) +3d_k +8\delta -\Delta_k\\
&\leq \CH^1(\Gamma)  +11\delta -\Delta_k.
\end{align*}
By \eqref{eq:a4_2},
we have
\[\Delta_k\leq 12\delta < 2^{-k}.\]
\end{proof}

\begin{corollary}\label{co:5_2}
Let $\Gamma$ be a curve with the property described in property 5 of Theorem \ref{th:3_2}.
Then the length of $\Gamma$ -- $\CH^1(\Gamma)$ is not computable.
\end{corollary}
\begin{proof}
We prove the contrapositive.
Let $\Gamma$ be a curve with a computable parametrization with a computable length $\CH^1(\Gamma)$.
Then by Theorem \ref{th:5_1}, we can use the Turing machine that computes $\CH^1(\Gamma)$
as the oracle in the statement of Theorem \ref{th:5_1} and obtain a Turing machine that
computes the constant speed parametrization of $\Gamma$.
Therefore, $\Gamma$ does not have the property described in item 5 of Theorem \ref{th:3_2}.
\end{proof}
\end{section} 
\begin{section}{Conclusion}\label{se:6}
As we have noted, Ko \cite{Ko95} has proven the
existence of computable curves with finite,
but uncomputable lengths, and the curve $\mathbf{\Gamma}$
of our main theorem is one such curve.
In the recent paper \cite{Gu:PCC}, we have given
a precise characterization of those points in
$\R^n$ that lie on computable curves of
finite length. With these things in mind,
we pose the following.

{\bf Question.} Is there a point $x\in\R^n$ such that $x$ lies on a computable
curve of finite length but not on any computable curve of computable length?
\end{section} 

\begin{ack}
We thank anonymous referees for their valuable comments.
\end{ack}

\bibliographystyle{abbrv}
\bibliography{dim,rbm,main,random,dimrelated}

\end{document}